\documentclass[a4paper,fleqn]{cas-dc}
\usepackage[numbers]{natbib}
\usepackage{algorithm} 
\usepackage{algcompatible}
\usepackage{algpseudocode}

\usepackage{caption} 
\usepackage{physics}
\usepackage{amsthm}
\usepackage{mdframed}
\floatname{algorithm}{Protocol}
\usepackage{tabularx}
\usepackage[most]{tcolorbox}
\usepackage{float}
\usepackage{subcaption}
\usepackage{enumitem}
\usepackage{multirow}
\usepackage{capt-of}   
\usepackage{hyperref}
\usepackage{needspace}
\usepackage[dvipsnames]{xcolor}
\definecolor{deepgreen}{rgb}{0, 0.6, 0.2}
\definecolor{MyPurple}{RGB}{120, 30, 120}
\hypersetup{
    colorlinks=true,    
    linkcolor=blue,      
    citecolor=deepgreen, 
    urlcolor=MyPurple,   
    breaklinks=true     
}

\theoremstyle{definition}

\newtheorem{theorem}{Theorem}
\newtheorem{remark}{Remark}
\begin{document}
\let\WriteBookmarks\relax
\def\floatpagepagefraction{1}
\def\textpagefraction{.001}

\shorttitle{QTPSI}    
\shortauthors{Zixian Gong et~al.}  
\title [mode = title]{Quantum Multi-Party Threshold Private Set Intersection with Explicit Cardinality Testing}

\author[1]{Zixian Gong}[orcid={0009-0005-7059-5040}]
\cormark[1]
\author[1]{Kun Tian}
\author[1]{Yi Zhang}
\author[2]{Fengxia Liu}
\ead{ArtoriasGong@ruc.edu.cn}

\affiliation[1]{organization={School of Mathematics, Renmin University of China},
                city={Beijing},
                postcode={100872}, 
                country={P. R. China}}

\affiliation[2]{organization={Great Bay University},
                city={DongGuan},
                postcode={523808}, 
                country={P. R. China}}
\cortext[1]{Corresponding author}
\begin{abstract}
Threshold private set intersection (TPSI) allows parties to reveal their intersection only when its cardinality reaches a prescribed threshold. Existing quantum TPSI protocols typically rely on a third party (TP) to interpret the final results,  which deviates from the cardinality-testing paradigm of TPSI. In this paper, we propose a quantum multiparty TPSI protocol with explicit cardinality testing. 
Our protocol develops a rotation-based quantum construction in which single-photon sequences are sequentially processed through participant-side data rotations, TP--participant masking rotations, and correlated aggregate rotations. 
This design produces hidden-label measurement vectors: TP can complete the final measurement, but cannot interpret the semantic meaning of the outcomes. 
Based on these hidden measurements, we further realize the threshold decision through an oblivious linear evaluation (OLE)-based inner product procedure and a lightweight garbled circuit, revealing only
\(\mathbf 1[|\bigcap_i X_i|\ge \tau]\)
before conditional intersection reconstruction. 
We prove the correctness and security of the proposed protocol, and further validate its feasibility through quantum-circuit simulations implemented on the IBM \textsf{Qiskit} platform.
\end{abstract}
\maketitle
\section{Introduction}
\noindent Private set intersection (PSI) is a fundamental primitive in secure multiparty computation (SMC). It enables a set of parties, each holding a private set \(X_i\), to compute their common intersection \(\bigcap_i X_i\) without revealing any information about elements outside the intersection~\cite{Mea86,FNP04}. Owing to this privacy-preserving matching capability, PSI has become a key building block for a wide range of applications including contact tracing~\cite{abc20}, genomic data analysis~\cite{Shen18}, and ad conversion tracking~\cite{Mihaela17}.

Driven by the diversity of PSI applications, a variety of functional extensions have been developed to satisfy different privacy and utility requirements. Representative examples include PSI cardinality protocols, which reveal only the size of the intersection~\cite{Debnath15}; circuit PSI, which delivers the intersection in secret-shared form for subsequent secure computation~\cite{Huang12}; and fuzzy PSI, which supports similarity-based matching~\cite{Uzun21}. Among these variants, TPSI introduces a threshold-gated disclosure mechanism, namely,
\(
    \text{reveal } \bigcap_i X_i \text{only if}
    \left| \bigcap_i X_i \right| \ge \tau .
\)
This functionality is useful in practical applications such as privacy-preserving ridesharing~\cite{Hallgren17,Mohanty24} and matchmaking~\cite{ZhaoChow18}. Beyond such application-driven formulations, TPSI has also been studied as a distinct cryptographic primitive with dedicated communication-efficiency and multiparty constructions~\cite{GhoshSimkin19,BMRR21}.

The emergence of Shor's algorithm~\cite{Shor97} has shown that many cryptographic schemes relying on number-theoretic hardness assumptions are vulnerable in the presence of quantum computers. While some ongoing researches focuse on post-quantum cryptography (PQC) as an interim measure, Shi \emph{et al.}~\cite{QPSI_Shi16} proposed the first two-party quantum PSI protocol. \cite{QPSI_Zhang20} considered a three-party setting with a third party (TP) and realized quantum PSI-CA and PSU-CA based on GHZ states. \cite{QPSI_M23} further extended quantum PSI to multi-party setting (MP-QPSI) using single photons and unitary operations. \cite{QPSI_Huang24} achieved MP-QPSI via rotation operations. More recently, the threshold functionality has also been introduced into the quantum setting, leading to quantum TPSI based on rotation operations~\cite{Mohanty24,LI26} and quantum homomorphic encryption (QHE)~\cite{Wang26}.

Although several quantum TPSI protocols have been proposed \cite{Mohanty24, LI26, Wang26}, their threshold mechanisms do not entirely fully align with the ideal classical cardinality testing paradigm. 
In these constructions, a TP is typically introduced to assist the quantum procedure, perform the final measurement or result processing, recover the intersection cardinality, and then compare it with the public threshold \(\tau\). 
Consequently, TP simultaneously acts as a quantum-resource provider, a protocol participant, and an interpreter of the final threshold-related outcome. 
This creates both a concentration of authority and a functionality mismatch: recovering \(\left|\bigcap_{i=1}^{n}X_i\right|\) and then comparing it with \(\tau\) is strictly more revealing than directly realizing the testing predicate \(\mathbf 1\!\left[\left|\bigcap_{i=1}^{n}X_i\right|\ge \tau\right].\)  This gap motivates us to design a quantum TPSI protocol that decouples the measurement or result-processing role from the semantic interpretation role:
\begin{itemize}[leftmargin=12pt, itemsep=0pt, parsep=0pt, topsep=2pt]
    \item \textbf{Functionality gap and rotation-based hidden-label quantum design.}
    We identify that existing quantum TPSI schemes rely on TP-side cardinality recovery followed by threshold comparison. 
    To restore an explicit cardinality-testing paradigm, we develop a rotation-based quantum construction over single-photon sequences, combining participant-side data rotations, TP--participant masking rotations, and correlated aggregate rotations to produce hidden-label measurement outcomes that TP cannot directly interpret.

    \item \textbf{Cardinality testing for hidden quantum measurements.}
    We convert the hidden measurement outputs into a masked label-consistency test and realize the threshold decision through oblivious linear evaluation-based inner-product sharing and a lightweight garbled circuit. 
    The protocol outputs only \(\mathbf 1[|\bigcap_i X_i|\ge \tau]\) before conditional reconstruction, without revealing the exact cardinality.

    \item \textbf{Security, simulation, and comparison.}
    We prove correctness and security against outside eavesdroppers, TP, and colluding participants, and discuss anchor-based tamper detection. 
    We also validate the quantum phase via \textsf{Qiskit} simulation and compare our protocol with representative TPSI schemes.
\end{itemize}

\emph{Outline.} The remainder of this paper is organized as follows. 
Section~\ref{sectionMo} discusses the cardinality testing in quantum TPSI. 
Section~\ref{sectionCT} presents the realization of cardinality testing for hidden-label measurements. 
The proposed MP-QTPSI protocol is described in Section~\ref{sectionProtocol} along with correctness and security analysis in Section~\ref{sectionAnalysis}. Simulation and comparisons are given in Section~\ref{sectionSimulation}. 
Finally, Section~\ref{sectionConclusion} concludes the paper.
\section{Cardinality Testing for Quantum TPSI}\label{sectionMo}
\noindent\emph{\textbf{Cardinality testing in TPSI.}}
TPSI is intended to realize a conditional disclosure rule: the intersection is revealed only when its cardinality reaches a public threshold. 
Formally, the protocol should first determine
\[
\qquad\mathsf{flag}
=
\mathbf 1\!\left[
\left|\bigcap_{i=1}^{n}X_i\right|\ge \tau
\right],
\]
and release \(\bigcap_{i=1}^{n}X_i\) only if \(\mathsf{flag}=1\) \cite{GhoshSimkin19,BMRR21,GhoshSimkin23}. 
Thus, cardinality testing is a distinct functionality from cardinality disclosure and then comparison:
\[
\quad\mathbf 1\!\left[
\left|\bigcap_{i=1}^{n}X_i\right|\ge \tau
\right]
\overset{Functionality}{\neq}
\left|\bigcap_{i=1}^{n}X_i\right|\ge \tau.
\]the former reveals only a one-bit threshold predicate, whereas the latter exposes an additional statistic of the private sets.

\emph{\textbf{Limitation of TP-centered quantum TPSI.}}
In existing quantum TPSI protocols~\cite{Mohanty24,LI26,Wang26}, TP is typically responsible for quantum-state preparation or result processing, final measurement, and threshold-related interpretation. 
In particular, TP first derives the intersection cardinality, compares it with \(\tau\), and then decides whether to release the intersection related result.  Hence, TP's measurement role and semantic interpretation role are coupled together.  This makes the realized functionality closer to PSI-CA and comparison, rather than direct cardinality testing.  To recover the TPSI paradigm without revealing the exact cardinality to TP, the final quantum measurements must first be made semantically uninterpretable to TP.

\emph{\textbf{From Quantum Measurements to Hidden-Label Vectors.}}
We address this issue by introducing a participant-side hidden flip vector. The flip at position \(t\) is jointly realized through correlated aggregate rotations of all participants, so TP can perform the final measurement but cannot determine which deterministic label corresponds to an intersection-consistent outcome. 
The participant side keeps a synchronized reference-label vector, which records the correct post-flip interpretation of each position.

Since our construction is also based on rotations of angle $\pi/n$, positions held by only a subset of participants may yield probabilistic same/opposite outcomes, TP measures \(\ell\) repeated photon sequences and compresses them into two deterministic-label vectors \(z^{\mathsf S}\) and \(z^{\mathsf O}\). 
In addition, positive and negative anchor sets are mixed with the real domain before the position-hiding map, which further obscures TP's structural view of deterministic positions and later supports an anchor-consistency check. 
Consequently, after the quantum phase, TP holds only the hidden-label vectors \(z^{\mathsf S}\) and \(z^{\mathsf O}\), while the participant side retains the information needed to interpret them. In this way, the quantum phase no longer outputs a cardinality-interpretable result to TP, but only a semantically blinded measurement interface for the subsequent test.
The remaining task is therefore reduced to testing the threshold condition over hidden-label measurements without exposing the labels or the exact cardinality.

\section{Realizing Cardinality Testing for Hidden-Label Measurements}\label{sectionCT}
\noindent After the quantum phase discussed above, TP does not obtain an interpretable intersection vector.  Instead, TP only holds the hidden-label measurement vectors , whose labels cannot be mapped to intersection and non-intersection positions without the participant-side information. Therefore, the remaining task is to test the threshold condition over these hidden-label vector without revealing the measurement labels, the reference labels, and the exact cardinality. We realize this task by reducing it to a masked label-consistency test and instantiating the required inner products using oblivious linear evaluation (OLE).

\noindent\begin{minipage}{\linewidth}
\begin{tcolorbox}[
    width=\linewidth,
    colback=white,
    colframe=black,
    boxrule=0.8pt,
    arc=0pt,
    top=4pt,
    bottom=4pt,
    left=6pt,
    right=6pt,
    boxsep=0pt
]
\small
\noindent\underline{\(\mathcal F_{\mathsf{OLE}}^{p}\): Oblivious Linear Evaluation (OLE)}\vspace{0.8ex}

\noindent\textbf{Parameters:} A finite field \(\mathbb F_p\).\\
\textbf{Inputs:} Alice inputs \(x\in\mathbb F_p\). Bob inputs \(u,v\in\mathbb F_p\).\\
\textbf{Output:} Alice learns \(z=ux+v\in\mathbb F_p\). Bob learns \(\bot\).

\vspace{1.8ex}

\noindent\underline{\(\mathcal F_{\mathsf{VOLE}}^{p}\): Vector Oblivious Linear Evaluation (VOLE)}\vspace{0.8ex}

\noindent\textbf{Parameters:} A finite field \(\mathbb F_p\), and vector length \(L\).\\
\textbf{Inputs:} Alice inputs \(x\in\mathbb F_p\). Bob inputs \(\mathbf u,\mathbf v\in\mathbb F_p^{L}\).\\
\textbf{Output:} Alice learns \(\mathbf z=\mathbf{u}x+\mathbf v\in\mathbb F_p^{L}\). Bob learns \(\bot\).
\end{tcolorbox}
\captionof{figure}{Ideal functionalities for OLE and VOLE.}
\label{olevole-functionality}
\end{minipage}
\subsection{Oblivious Inner Product (OIP) from OLE}
\noindent\emph{\textbf{OLE and Vector OLE Functionalities.}} OLE can be viewed as the linear case of Oblivious Polynomial Evaluation (OPE) \cite{OPE99} that enables a receiver to compute a linear combination of the sender's inputs. Vector OLE (VOLE) generalizes this primitive to vectors, as illustrated in Fig.~\ref{olevole-functionality}. VOLE significantly reduces the overhead of batch OLE generation, making it a cornerstone for efficient cryptographic protocols such as zero-knowledge proofs and PSI \cite{BCGI18, Weng21, RS21}. They can be instantiated using different cryptographic techniques, such as OT-based extensions, function secret sharing (FSS) with distributed point functions \cite{BCGI18}. In particular, for our quantum setting, it can be instantiated from lattice-based assumptions such as RLWE \cite{Baum22}.

\emph{\textbf{Constructing OIP.}} For our cardinality testing subroutine, the useful interface is not OLE itself, but an oblivious inner product (OIP) functionality. 
Given two private vectors, OIP outputs additive shares of their inner product as shown in \ref{oip-functionality}, which is exactly the form needed later for Hamming distance shares. Concretely, let Alice hold \(\mathbf a=(a_0,\ldots,a_{L-1})\in\mathbb F_p^L\) and Bob hold \(\mathbf b=(b_0,\ldots,b_{L-1})\in\mathbb F_p^L\). 
For each coordinate \(t\), Alice inputs \(a_t\) to OLE, while Bob inputs \((b_t,r_t)\) for a random \(r_t\in\mathbb F_p\). 
Alice obtains \(a_tb_t+r_t\), and Bob keeps \(-r_t\) as his local share. 
After summing over all coordinates, Alice and Bob obtain additive shares
\(
s_A=\sum_{t=0}^{L-1}(a_tb_t+r_t),
s_B=-\sum_{t=0}^{L-1}r_t,
\)
which satisfy $s_A+s_B=\langle \mathbf a,\mathbf b\rangle \pmod p$.

\noindent\begin{minipage}{\linewidth}
\begin{tcolorbox}[
    width=\linewidth,
    colback=white,
    colframe=black,
    boxrule=0.8pt,
    arc=0pt,
    top=4pt,       
    bottom=4pt,    
    left=6pt,      
    right=6pt,     
    boxsep=0pt     
]
\small
\noindent\underline{\(\mathcal F_{\mathsf{OIP}}^{p}\): Oblivious Inner Product (OIP)}\vspace{0.8ex}

\noindent\textbf{Parameters:} A finite field \(\mathbb F_p\), and vector length \(L\).\\
\noindent\textbf{Inputs:} Alice inputs \(\mathbf a\in\mathbb F_p^L\). Bob inputs \(\mathbf b\in\mathbb F_p^L\).\\
\noindent\textbf{Output:} Alice receives \(s_A\in\mathbb F_p\), and Bob receives \(s_B\in\mathbb F_p\), such that \(s_A+s_B=\langle \mathbf a,\mathbf b\rangle \pmod p\).
\end{tcolorbox}
\captionof{figure}{Ideal functionality for oblivious inner product (OIP).}
\label{oip-functionality}
\end{minipage}
\medskip

Since the consistency statistics are linear combinations of inner products between TP-held label vectors and participant-held masks, OIP provides a natural interface for realizing the classical testing layer.
\subsection{Masked Label-Consistency Sharing and Threshold Decision}\label{sectionSharing}
\begin{table*}[t]
\centering
\small
\renewcommand{\arraystretch}{1.2} 
\caption{Summary of notation.}
\label{notation}
\begin{tabularx}{\linewidth}{l | X l | X}
\hline
\textbf{Symbol} & \textbf{Description} & \textbf{Symbol} & \textbf{Description} \\
\hline
$P_i$ & The $i$-th participant, where $i\in[n]$. & $\boldsymbol{\Theta}$ & Semantic-flip angle vector with $\Theta_t=b_t\pi$. \\

$M$ & Augmented domain size. & $\boldsymbol{\Delta}_i$ & Rotation-share vector of $P_i$. \\

$X_i$ & Private set held by $P_i$. & $T_i$ & Pairwise TP--$P_i$ rotation-mask vector. \\

$Y_i$ & Length-$M$ indicator vector encoded by $P_i$. & $\vartheta_0$ & TP-local initial rotation vector. \\

$\tau$ & Threshold for revealing the intersection. & $\vartheta_{i,t}$ & Data-dependent rotation angle of $P_i$ at position $t$. \\

$\ell$ & Repetition number. & $z^{\mathsf S},z^{\mathsf O}$ & Indicator vectors obtained by TP. \\

$K_{\mathsf P}$ & Participant-only master key. & $\rho$ & Participant-side reference label vector. \\

$k$ & Participant-side multiplicative hiding key. & $m_{\mathcal U},m_{\mathcal A}$ & Selector masks for real and anchor positions. \\

$\mathbf b$ & Hidden label-flip vector. & $d_{\mathcal U},d_{\mathcal A}$ & Inconsistency counts. \\
\hline
\end{tabularx}
\end{table*}

\noindent\emph{\textbf{Masked label-consistency counting.}}
We now instantiate the OIP interface for the deterministic measurement outcomes obtained in the quantum phase. 
Let \(z^{\mathsf S},z^{\mathsf O}\in\{0,1\}^M\) be TP's indicator vectors, where \(z_t^{\mathsf S}=1\) if the \(\ell\) measurement results at position \(t\) are all-same as initial basis, and \(z_t^{\mathsf O}=1\) if they are all-opposite compared to initial basis. 
For mixed outcomes, \(z_t^{\mathsf S}=z_t^{\mathsf O}=0\). 
The participant side holds a reference label vector \(\rho\in\{0,1\}^M\), where \(\rho_t=1\oplus b_t\) represents the deterministic label expected for an intersection position after the hidden flip \(b_t\). It also holds two selector vectors \(m_{\mathcal U},m_{\mathcal A}\in\{0,1\}^M\), selecting the hidden real-domain positions and the hidden anchor positions, respectively.

For each position \(t\), define the label-consistency indicator as \(\chi_t=(1-\rho_t)z_t^{\mathsf S}+\rho_t z_t^{\mathsf O}\). 
Thus, \(\chi_t=1\) exactly when the deterministic measurement label at position \(t\) agrees with the participant-side reference label; mixed outcomes contribute \(0\). 
For \(R\in\{\mathcal U,\mathcal A\}\), define the consistency count \(C_R=\sum_{t=0}^{M-1}m_{R,t}\chi_t\) and the corresponding inconsistency count \(d_R=|R|-C_R=|R|-\sum_{t=0}^{M-1}m_{R,t}\chi_t\). 
Equivalently,
\[
d_R
=
|R|
-
\left\langle z^{\mathsf S},m_R\odot(1-\rho)\right\rangle
-
\left\langle z^{\mathsf O},m_R\odot\rho\right\rangle .
\]
Here, \(d_{\mathcal U}\) counts real-domain positions that are not confirmed as intersection-consistent, while \(d_{\mathcal A}\) checks whether the anchor positions satisfy their prescribed labels.

\noindent
\begin{minipage}{\linewidth}
    \centering

    \includegraphics[width=0.9\linewidth]{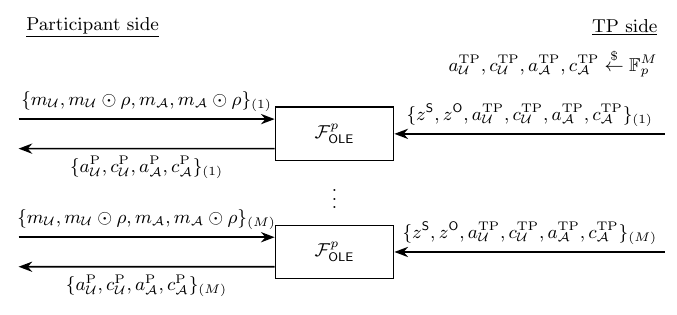}
    \captionof{figure}{Using OLE for additive shares.}
    \label{figOLE}

    \vspace{1.2em}

    \includegraphics[width=0.8\linewidth]{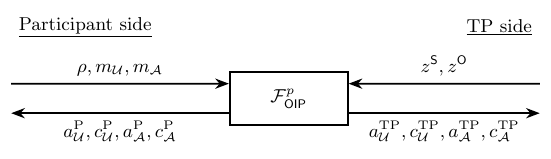}
    \captionof{figure}{OIP Constructed from OLE.}
    \label{figOIP}
\end{minipage}

\emph{\textbf{Consistency-share generation.}}
TP and the participant side invoke \hyperref[oip-functionality]{\(\mathcal F_{\mathsf{OIP}}^p\)} to obtain additive shares of
\(\langle z^{\mathsf S},m_{\mathcal U}\odot(1-\rho)\rangle\), 
\(\langle z^{\mathsf O},m_{\mathcal U}\odot\rho\rangle\), 
\(\langle z^{\mathsf S},m_{\mathcal A}\odot(1-\rho)\rangle\), and 
\(\langle z^{\mathsf O},m_{\mathcal A}\odot\rho\rangle\).
The OIP interface can be instantiated from coordinate-wise OLE calls, as illustrated in Fig.~\ref{figOLE} and ~\ref{figOIP}.

For \(R\in\{\mathcal U,\mathcal A\}\), write
\(
a_R^{\mathrm{TP}}+a_R^{\mathrm P}
=
\left\langle z^{\mathsf S},m_R\odot(1-\rho)\right\rangle,\)
\(c_R^{\mathrm{TP}}+c_R^{\mathrm P}
=
\left\langle z^{\mathsf O},m_R\odot\rho\right\rangle
\pmod p .
\)
Then TP computes \(\delta_R^{\mathrm{TP}}=-a_R^{\mathrm{TP}}-c_R^{\mathrm{TP}}\pmod p\), while the participant side computes \(\delta_R^{\mathrm P}=|R|-a_R^{\mathrm P}-c_R^{\mathrm P}\pmod p\). 
It follows that \(\delta_R^{\mathrm{TP}}+\delta_R^{\mathrm P}=d_R\pmod p\). 

\emph{Anchor check and threshold decision.}
After the sharing step, TP holds \((\delta_{\mathcal U}^{\mathrm{TP}},\delta_{\mathcal A}^{\mathrm{TP}})\), while the participant side holds \((\delta_{\mathcal U}^{\mathrm P},\delta_{\mathcal A}^{\mathrm P})\). 
The final Boolean predicate is evaluated by a lightweight garbled circuit \(C_{\mathrm{CT}}\), following the standard secure-computation approach for privately evaluating low-depth decision logic \cite{Yao86}, which reconstructs \(d_{\mathcal U}=(\delta_{\mathcal U}^{\mathrm{TP}}+\delta_{\mathcal U}^{\mathrm P})\bmod p\) and \(d_{\mathcal A}=(\delta_{\mathcal A}^{\mathrm{TP}}+\delta_{\mathcal A}^{\mathrm P})\bmod p\), and outputs only
\[
\mathsf{flag}=1
\iff
(d_{\mathcal A}=0)\wedge(d_{\mathcal U}\le |\mathcal U|-\tau).
\]
The condition \(d_{\mathcal A}=0\) verifies the anchor positions, while \(d_{\mathcal U}\le |\mathcal U|-\tau\) means that at least \(\tau\) real-domain positions are confirmed as intersection-consistent. 
The complete subprotocol is shown in Fig.~\ref{subprotocol}. 
By the privacy of garbled-circuit evaluation, the parties learn only the Boolean predicate value, not the explicit inconsistency counts \(d_{\mathcal U},d_{\mathcal A}\). 
Thus, the subprotocol realizes cardinality testing as \(\mathbf 1[|\bigcap_i X_i|\ge \tau]\), rather than exposing \(|\bigcap_i X_i|\) and then performing a comparison with \(\tau\).

\noindent\begin{minipage}{\linewidth}
\begin{tcolorbox}[
    width=\linewidth,
    colback=white,
    colframe=black,
    boxrule=0.8pt,
    arc=0pt,
    left=6pt,right=6pt,top=5pt,bottom=5pt,
    boxsep=0pt
]
\small
\noindent\textbf{Protocol \(\Pi_{\mathrm{CT}}^{\tau}\): Cardinality Testing from OIP and GC.}\vspace{1ex}

\noindent\underline{\textbf{Parameters:}} Augmented length \(M\), finite field \(\mathbb F_p\) with \(p>2M\), original domain size \(|\mathcal U|\), anchor-domain size \(|\mathcal A|\), and threshold \(\tau\).\vspace{0.8ex}

\noindent\underline{\textbf{Inputs:}} TP holds \(z^{\mathsf S},z^{\mathsf O}\in\{0,1\}^M\). The participant side holds \(\rho,m_{\mathcal U},m_{\mathcal A}\in\{0,1\}^M\), where \(\rho\) is the reference label vector, \(m_{\mathcal U}\) selects the hidden real-domain positions, and \(m_{\mathcal A}\) selects the hidden anchor positions.\vspace{0.8ex}

\noindent\underline{\textbf{OIP Share:}} TP and the participant side invoke \hyperref[oip-functionality]{\(\mathcal F_{\mathsf{OIP}}^{p}\)} four times to obtain additive shares of
\(\langle z^{\mathsf S},m_{\mathcal U}\odot(1-\rho)\rangle\),
\(\langle z^{\mathsf O},m_{\mathcal U}\odot\rho\rangle\),
\(\langle z^{\mathsf S},m_{\mathcal A}\odot(1-\rho)\rangle\), and
\(\langle z^{\mathsf O},m_{\mathcal A}\odot\rho\rangle\).
All the following share equalities are taken modulo $p$:
\vspace{-1.5ex}
\[
\begin{aligned}
a_{\mathcal U}^{\mathrm{TP}}+a_{\mathcal U}^{\mathrm P}
&= \langle z^{\mathsf S},m_{\mathcal U}\odot(1-\rho)\rangle,
&
c_{\mathcal U}^{\mathrm{TP}}+c_{\mathcal U}^{\mathrm P}
&= \langle z^{\mathsf O},m_{\mathcal U}\odot\rho\rangle,\\
a_{\mathcal A}^{\mathrm{TP}}+a_{\mathcal A}^{\mathrm P}
&= \langle z^{\mathsf S},m_{\mathcal A}\odot(1-\rho)\rangle,
&
c_{\mathcal A}^{\mathrm{TP}}+c_{\mathcal A}^{\mathrm P}
&= \langle z^{\mathsf O},m_{\mathcal A}\odot\rho\rangle.
\end{aligned}
\]

\vspace{-0.5ex}
\noindent\underline{\textbf{Consistency Share:}} TP locally computes
\(\delta_{\mathcal U}^{\mathrm{TP}}=-a_{\mathcal U}^{\mathrm{TP}}-c_{\mathcal U}^{\mathrm{TP}}\)
and
\(\delta_{\mathcal A}^{\mathrm{TP}}=-a_{\mathcal A}^{\mathrm{TP}}-c_{\mathcal A}^{\mathrm{TP}}\).
The participant side locally computes:
\vspace{-1.5ex}
\[
\begin{aligned}
\qquad\qquad\delta_{\mathcal U}^{\mathrm P}
&= |\mathcal U|-a_{\mathcal U}^{\mathrm P}-c_{\mathcal U}^{\mathrm P},\\
\delta_{\mathcal A}^{\mathrm P}
&= |\mathcal A|-a_{\mathcal A}^{\mathrm P}-c_{\mathcal A}^{\mathrm P}\pmod p.
\end{aligned}
\]

\vspace{-0.5ex}
\noindent\underline{\textbf{Threshold GC:}} TP and the participant side evaluate a garbled circuit \(C_{\mathrm{CT}}\) with inputs \((\delta_{\mathcal U}^{\mathrm{TP}},\delta_{\mathcal A}^{\mathrm{TP}})\) from TP and \((\delta_{\mathcal U}^{\mathrm P},\delta_{\mathcal A}^{\mathrm P})\) from the participant side. The circuit reconstructs \(d_{\mathcal U}=(\delta_{\mathcal U}^{\mathrm{TP}}+\delta_{\mathcal U}^{\mathrm P})\bmod p\) and \(d_{\mathcal A}=(\delta_{\mathcal A}^{\mathrm{TP}}+\delta_{\mathcal A}^{\mathrm P})\bmod p\), and outputs:
\vspace{-1.5ex}
\[
\mathsf{flag}=1 \iff (d_{\mathcal A}=0)\wedge(d_{\mathcal U}\le |\mathcal U|-\tau).
\]

\vspace{-0.5ex}
\noindent\underline{\textbf{Output:}} Both sides learn only \(\mathsf{flag}\). Neither side learns \(d_{\mathcal U}\), \(d_{\mathcal A}\), nor the other side's private vectors.
\end{tcolorbox}
\vspace{-1em}
\captionof{figure}{Masked Hidden-label Cardinality Testing Protocol.}
\label{subprotocol}
\end{minipage}

\begin{figure*}[t]  
    \centering
    \includegraphics[width=0.9\linewidth]{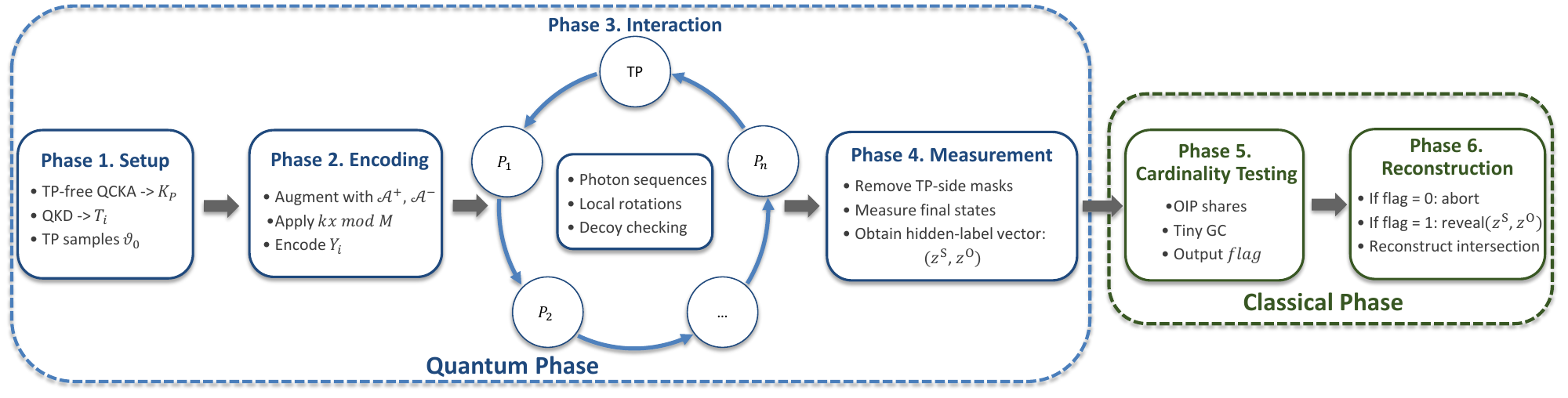}
    \caption{Flowchart.}
    \label{figFlowchart}
\end{figure*}

\section{Proposed Quantum TPSI Protocol}\label{sectionProtocol}
\noindent We now present the proposed  MP-QTPSI protocol. 
The protocol involves \(n\) participants \(P_1,\ldots,P_n\) and a semi-honest third party TP. 
Each participant \(P_i\) holds a private set \(X_i=\{x_i^1,\ldots,x_i^{\ell_i}\}\subseteq\mathcal U\). 
The goal is to reveal \(\bigcap_{i=1}^n X_i\) only when its cardinality is at least the threshold \(\tau\), and otherwise reveal no intersection element.

\indent The protocol uses two anchor sets \(\mathcal A^+\) and \(\mathcal A^-\) in addition to the original universe \(\mathcal U\). After the position-hiding map, these anchors are mixed with the real-domain positions and are later used to define the reference labels and to support the anchor-consistency check. Since TP does not know the hidden anchor locations and labels, this check can detect inconsistent or tampered measurement output with a certain probability.

\indent In our protocol, TP is used as a quantum-state preparer and measurement party, but it is not allowed to collude with any participant. Different from existing TP-centered quantum TPSI protocols, TP does not directly obtain interpretable intersection indices or the exact intersection cardinality. Instead, the quantum interaction produces a hidden-label measurement vector, and the threshold decision is made through the masked cardinality-testing subprotocol described in Section~\ref{sectionCT}. The overall workflow is illustrated in Fig.~\ref{figFlowchart}.

\begin{algorithm*}[!tp]
\small
\caption{Quantum Multi-party Threshold PSI with Cardinality Testing}
\label{tpsi}
\begin{algorithmic}[1]
\Require Private sets $X_i\subseteq\mathcal U$, threshold $\tau$, anchor sets $\mathcal A^+,\mathcal A^-$, repetition number $\ell$.
\Ensure Reveal $\bigcap_{i=1}^n X_i$ iff its cardinality is at least $\tau$.

\Statex \hspace{-\algorithmicindent} \underline{\textbf{Phase 1: Setup.}}
\State Let $\widetilde{\mathcal U}=\mathcal U\cup\mathcal A^+\cup\mathcal A^-$ and $M=|\widetilde{\mathcal U}|$.
\State Participants run TP-free QCKA to obtain $K_{\mathsf P}$ and derive $k\in\mathbb Z_M^*$, $\mathbf b=(b_0,\ldots,b_{M-1})\in\{0,1\}^M$, $\boldsymbol{\Theta}=(\Theta_0,\ldots,\Theta_{M-1})$ with $\Theta_t=b_t\pi$, and rotation-share vectors $\boldsymbol{\Delta}_i=(\Delta_{i,0},\ldots,\Delta_{i,M-1})$ satisfying $\sum_{i=1}^n\Delta_{i,t}\equiv\Theta_t\pmod{2\pi}$ for all $t$.
\State For each $i\in[n]$, TP and $P_i$ run QKD to derive $T_i=(T_{i,0},\ldots,T_{i,M-1})$; TP samples $\vartheta_0=(\vartheta_{0,0},\ldots,\vartheta_{0,M-1})$.

\Statex \hspace{-\algorithmicindent} \underline{\textbf{Phase 2: Encoding.}}
\State Compute the hidden domains $k\mathcal U$, $k\mathcal A^+$, and $k\mathcal A^-\pmod M$.
\State Define $m_{\mathcal U}=\mathbf 1_{k\mathcal U}$, $m_{\mathcal A}=\mathbf 1_{k\mathcal A^+\cup k\mathcal A^-}$, and $\rho_t=1\oplus b_t$ for $t\in k\mathcal U\cup k\mathcal A^+$, $\rho_t=b_t$ for $t\in k\mathcal A^-$.
\For{$i=1$ to $n$}
    \State $P_i$ sets $\widetilde X_i=X_i\cup\mathcal A^+$ and encodes $Y_i=\mathbf 1_{\{kx\bmod M:x\in\widetilde X_i\}}\in\{0,1\}^M$.
\EndFor

\Statex \hspace{-\algorithmicindent} \underline{\textbf{Phase 3: Interaction and Measurement.}}
\State TP prepares $\ell$ length-$M$ photon sequences, applies $R_y(\vartheta_{0,t})$ to each position $t$, inserts decoys, and sends the sequence to $P_1$.
\For{$i=1$ to $n$}
    \State $P_i$ checks decoys and applies $R_y(\vartheta_{i,t}+T_{i,t}+\Delta_{i,t})$ at each position $t$, where $\vartheta_{i,t}=\pi/n$ if $Y_{i,t}=1$ and $\vartheta_{i,t}=0$ otherwise.
    \State $P_i$ inserts fresh decoys and forwards the sequence to $P_{i+1}$.
\EndFor
\State TP checks decoys, applies $R_y(-\vartheta_{0,t}-\sum_{i=1}^n T_{i,t})$ at each position $t$, and measures to obtain \(z^{\mathsf S},z^{\mathsf O}\in\{0,1\}^M\).

\Statex \hspace{-\algorithmicindent} \underline{\textbf{Phase 4: Cardinality Testing.}}
\State TP and the participants invoke \hyperref[subprotocol]{\(\Pi_{\mathrm{CT}}^{\tau}\)} on $z^{\mathsf S},z^{\mathsf O}\in\{0,1\}^M$ and $(\rho,m_{\mathcal U},m_{\mathcal A},\tau)$.
\State \hyperref[subprotocol]{\(\Pi_{\mathrm{CT}}^{\tau}\)} outputs $\mathsf{flag}=1$ iff $d_{\mathcal A}=0$ and $d_{\mathcal U}\le |\mathcal U|-\tau$, where \(d_R=|R|-\sum_{t=0}^{M-1}m_{R,t}((1-\rho_t)z_t^{\mathsf S}+\rho_t z_t^{\mathsf O})\).

\Statex \hspace{-\algorithmicindent} \underline{\textbf{Phase 5: Intersection Reconstruction.}}
\If{$\mathsf{flag}=1$}
    \State TP broadcasts $z^{\mathsf S},z^{\mathsf O}\in\{0,1\}^M$, and each participant outputs $\mathcal I =
\left\{
k^{-1}t\bmod M:
t\in k\mathcal U,\ 
(1-\rho_t)z_t^{\mathsf S}+\rho_t z_t^{\mathsf O}=1
\right\}$.
\Else
    \State Abort.
\EndIf
\end{algorithmic}
\end{algorithm*}

\emph{\textbf{Setup.}} Let \(\widetilde{\mathcal U}=\mathcal U\cup\mathcal A^+\cup\mathcal A^-\) be the augmented universe with \(\{0,\ldots,M-1\}\), where \(\mathcal U=\{0,\ldots,q-1\}\) is the original real domain and the remaining positions are occupied by \(\mathcal A^+\) and \(\mathcal A^-\) and let \(M=|\widetilde{\mathcal U}|\). 
The participants first execute a TP-free QCKA protocol~\cite{Fu15,QCKA21} to establish a participant-only master key \(K_{\mathsf P}\). 
From \(K_{\mathsf P}\), they derive the multiplicative hiding key \(k\in\mathbb Z_M^*\), the hidden label-flip vector \(\mathbf b=(b_0,\ldots,b_{M-1})\in\{0,1\}^M\), the semantic-flip angle vector \(\boldsymbol{\Theta}=(\Theta_0,\ldots,\Theta_{M-1})\) with \(\Theta_t=b_t\pi\), and the rotation-share vectors \(\boldsymbol{\Delta}_i=(\Delta_{i,0},\ldots,\Delta_{i,M-1})\) satisfying \(\sum_{i=1}^n\Delta_{i,t}\equiv \Theta_t\pmod{2\pi}\) for every position \(t\). 
In addition, for each participant \(P_i\), TP and \(P_i\) run a QKD protocol~\cite{BB84} to establish a pairwise key, from which they derive the rotation-mask vector \(T_i=(T_{i,0},\ldots,T_{i,M-1})\). 
Finally, TP locally samples an initial blinding vector \(\vartheta_0=(\vartheta_{0,0},\ldots,\vartheta_{0,M-1})\).

\emph{\textbf{Encoding.}} Each participant augments its input as \(\widetilde X_i=X_i\cup\mathcal A^+\), while \(\mathcal A^-\) is inserted into none of the participants' sets. Using the hiding key \(k\), all positions are permuted by \(x\mapsto kx\bmod M\), giving the hidden domains \(k\mathcal U\), \(k\mathcal A^+\), and \(k\mathcal A^- \,mod\,M\).

The selector vectors \(m_{\mathcal U},m_{\mathcal A}\in\{0,1\}^M\) are defined according to this hidden partition where \(m_{\mathcal U}+m_{\mathcal A}=\mathbf 1\). Specifically, \(m_{\mathcal U,t}=1\) iff \(t\in k\mathcal U\), and \(m_{\mathcal A,t}=1\) iff \(t\in k\mathcal A^+\cup k\mathcal A^-\). 
Thus \(m_{\mathcal U}\) marks the positions originating from the real domain, while \(m_{\mathcal A}\) marks the anchor positions after the \(kx\bmod M\) permutation. 

The participants also define the reference label vector \(\rho\in\{0,1\}^M\). 
For positions expected to behave as matching positions, namely \(t\in k\mathcal U\cup k\mathcal A^+\), they set \(\rho_t=1\oplus b_t\); for negative-anchor positions \(t\in k\mathcal A^-\), they set \(\rho_t=b_t\), where \(b_t\) is the hidden label-flip bit used to blind the semantic meaning of the final measurement label at position \(t\). 
Finally, each participant encodes its hidden augmented set as indicator vector
\(
Y_i=\mathbf 1_{\{kx\bmod M:\,x\in\widetilde X_i\}}\in\{0,1\}^M,
\)
which is used in the rotation phase.

\emph{\textbf{Interaction and measurement.}} 
To reduce measurement uncertainty, TP prepares $\ell$ length-$M$ photon sequences, each following an identical state configuration where every qubit is randomly chosen from $\{|0\rangle, |1\rangle, |+\rangle, |-\rangle\}$.
Then, TP applies the initial blinding rotation \(R_y(\vartheta_{0})\), inserts decoy photons, and sends the protected sequences to \(P_1\). For \(i=1,\ldots,n\), participant \(P_i\) first checks the decoy photons with the previous sender and removes them. 
Then \(P_i\) applies the position-wise rotation \(R_y(\vartheta_{i,t}+T_{i,t}+\Delta_{i,t})\) to each photon at position \(t\), where \(\vartheta_{i,t}=\pi/n\) if \(Y_{i,t}=1\), and \(\vartheta_{i,t}=0\) otherwise. Afterwards, \(P_i\) inserts fresh decoy photons and forwards the protected sequences to \(P_{i+1}\) until \(P_n\) returns the final sequence to TP.

After the final decoy check, TP removes the rotations known to itself by applying \(R_y(-\vartheta_{0,t}-\sum_{i=1}^nT_{i,t})\) at each position \(t\). The remaining effective rotation is therefore \(\sum_{i=1}^n\vartheta_{i,t}+\sum_{i=1}^n\Delta_{i,t}=\sum_{i=1}^n\vartheta_{i,t}+\Theta_t=\sum_{i=1}^n\vartheta_{i,t}+b_t \pi\). TP measures the \(\ell\) sequences and derives two deterministic-outcome indicator vectors \(z^{\mathsf S},z^{\mathsf O}\in\{0,1\}^M\): \(z_t^{\mathsf S}=1\) if the \(\ell\) measurement outcomes at position \(t\) are all-same, \(z_t^{\mathsf O}=1\) if they are all-opposite, and \(z_t^{\mathsf S}=z_t^{\mathsf O}=0\) otherwise.

\emph{\textbf{Cardinality testing.}} Participants side and TP invoke the cardinality testing subprotocol \hyperref[subprotocol]{\(\Pi_{\mathrm{CT}}^{\tau}\)} described in Section~\ref{sectionCT}. The subprotocol first checks the anchor consistency through \(d_{\mathcal A}=0\), which verifies whether the returned measurement vector is consistent with the prescribed anchor positions. It then tests the real-domain threshold condition \(d_{\mathcal U}\le |\mathcal U|-\tau\), and outputs only the Boolean value \(\mathsf{flag}\). 

\emph{\textbf{Intersection Reconstruction.}} If \(\mathsf{flag}=1\), then both the anchor-consistency check and the cardinality testing have passed. TP broadcasts the complete label vectors \(z^{\mathsf S}\) and \(z^{\mathsf O}\) to all participants. 
Using the reference labels \(\rho\), the hidden real-domain selector \(m_{\mathcal U}\), and the inverse \(k^{-1}\bmod M\), each participant locally recovers the intersection as
\(
\mathcal I
=
\left\{
k^{-1}t\bmod M:
t\in k\mathcal U,\ 
(1-\rho_t)z_t^{\mathsf S}+\rho_t z_t^{\mathsf O}=1
\right\}.
\)
\section{Correctness and Security Analysis}\label{sectionAnalysis}
\subsection{Correctness Analysis}
\noindent In this subsection, we establish the correctness of the proposed protocol from the following three aspects:
\begin{enumerate}[label=(\arabic*),leftmargin=12pt, itemsep=0pt, parsep=0pt, topsep=2pt]
    \item TP obtains a hidden-label measurement vector that correctly encodes the intersection pattern;
    \item The \hyperref[subprotocol]{\(\Pi_{\mathrm{CT}}^{\tau}\)} outputs the correct threshold decision;
    \item Once the threshold condition is satisfied, the participants can correctly reconstruct the real intersection.
\end{enumerate}

\begin{theorem}[Correctness of intersection pattern]\label{theoremCo1}
The vectors \(z^{\mathsf S}\) and \(z^{\mathsf O}\) measured by TP correctly encode the intersection pattern over the augmented domain under the participant-side label flip vector \(\mathbf b\).
\end{theorem}
\begin{proof}
For each position \(t\), after all participants complete their rotations and TP removes the rotations known to itself, the remaining effective rotation is
\[
\vartheta_{0,t}
+\sum_{i=1}^{n}(\vartheta_{i,t}+T_{i,t}+\Delta_{i,t})
-\vartheta_{0,t}
-\sum_{i=1}^{n}T_{i,t}
=
\sum_{i=1}^{n}\vartheta_{i,t}+\Theta_t .
\]
Let \(r_t=\sum_{i=1}^{n}Y_{i,t}\). Since \(\vartheta_{i,t}=\pi/n\) if \(Y_{i,t}=1\) and \(\vartheta_{i,t}=0\) otherwise, the residual rotation at position \(t\) is
\(
\frac{r_t\pi}{n}+\Theta_t
=
\frac{r_t\pi}{n}+b_t\pi .
\)
We consider the following three cases.

\emph{Case 1: \(r_t=0\).}
No participant holds the encoded element at position \(t\), so the residual rotation is \(b_t\pi\). 
Hence the final measurement relation is all-same when \(b_t=0\), and all-opposite when \(b_t=1\).

\emph{Case 2: \(1\le r_t\le n-1\).}
Only a proper subset of participants holds the encoded element at position \(t\), and the residual rotation is
\(
\frac{r_t\pi}{n}+b_t\pi,
\)
which is neither \(0\) nor \(\pi\) modulo \(2\pi\). 
Therefore, the resulting sequence are in superposition states. 
Over the \(\ell\) repeated sequences, such a position is recorded as "mixed".

\emph{Case 3: \(r_t=n\).}
All participants hold the encoded element at position \(t\), so the residual rotation is
\(
\pi+b_t\pi
\equiv
(1\oplus b_t)\pi
\pmod{2\pi}.
\)
Hence the final measurement relation is all-opposite when \(b_t=0\), and all-same when \(b_t=1\).

Consequently, \(z_t^{\mathsf S}\) and \(z_t^{\mathsf O}\) encode the deterministic non-intersection and intersection labels after the participant-side flip \(b_t\), while partially matched positions are classified as "mixed". 
Thus, the vectors \(z^{\mathsf S}\) and \(z^{\mathsf O}\) correctly represent the intersection pattern over the augmented domain under the label-flip vector \(\mathbf b\).
\end{proof}
\begin{theorem}[Correctness of cardinality testing]\label{theoremCo2}
The subprotocol \hyperref[subprotocol]{\(\Pi_{\mathrm{CT}}^{\tau}\)} outputs the correct value of \(\mathsf{flag}\).
\end{theorem}
\begin{proof}
By the construction in Section~\ref{sectionSharing} and the steps of 
\hyperref[subprotocol]{\(\Pi_{\mathrm{CT}}^{\tau}\)}, TP and the participant side invoke 
\(\mathcal F_{\mathsf{OIP}}^{p}\) to obtain additive shares of the inner products required for the consistency counts over the real-domain and anchor positions. 
They then locally derive additive shares of \(d_{\mathcal U}\) and \(d_{\mathcal A}\), satisfying
\(d_R=\delta_R^{\mathrm{TP}}+\delta_R^{\mathrm P}\pmod p\) for \(R\in\{\mathcal U,\mathcal A\}\).

The tiny garbled circuit \(C_{\mathrm{CT}}\) reconstructs these values only inside the circuit and evaluates
\(
\mathsf{flag}=1
\iff
(d_{\mathcal A}=0)\wedge(d_{\mathcal U}\le |\mathcal U|-\tau).
\)
Hence, \(d_{\mathcal A}=0\) correctly performs the anchor check, while
\(d_{\mathcal U}\le |\mathcal U|-\tau\) correctly judge the threshold decision. 
Therefore, \hyperref[subprotocol]{\(\Pi_{\mathrm{CT}}^{\tau}\)} outputs the correct value of \(\mathsf{flag}\).
\end{proof}
\begin{theorem}[Correctness of reconstruction]\label{theoremCo3}
If \(\mathsf{flag}=1\), then the participants can recover the \(\bigcap_{i=1}^n X_i\) exactly.
\end{theorem}
\begin{proof}
If \(\mathsf{flag}=1\), TP broadcasts \(z^{\mathsf S}\) and \(z^{\mathsf O}\) to all participants. By the correctness in Theorem \ref{theoremCo1}, for each hidden real-domain position \(t\in k\mathcal U\) \(\chi_t=(1-\rho_t)z_t^{\mathsf S}+\rho_t z_t^{\mathsf O}=1\) holds exactly when position \(t\) corresponds to an element contained in every participant's encoded set. Hence,
\(
\{\,t\in k\mathcal U:\chi_t=1\,\}
=
k\!\left(\bigcap_{i=1}^n X_i\right).
\)

Since \(k\in\mathbb Z_M^*\), the map \(x\mapsto kx\bmod M\) is invertible. 
Therefore, each participant can apply \(k^{-1}\bmod M\) to the recovered hidden positions and obtain
\[
\left\{
k^{-1}t\bmod M:
t\in k\mathcal U,\ 
(1-\rho_t)z_t^{\mathsf S}+\rho_t z_t^{\mathsf O}=1
\right\}
=
\bigcap_{i=1}^n X_i .
\]
\end{proof}
\subsection{Security Analysis}
\noindent In this subsection, we analyze the security of the proposed protocol from the following three aspects:
\begin{enumerate}[label=(\arabic*),leftmargin=12pt, itemsep=0pt, parsep=0pt, topsep=2pt]
    \item Any outside eavesdropper cannot obtain information about the private sets of the participants;
    \item TP cannot obtain private information about the participants, including their private sets, the intersection indices, or the exact intersection cardinality;
    \item No participant can obtain extra information about the private sets of other participants even when colluding with other participants.
\end{enumerate}
together with a remark on the detection of tampering behavior when TP modifies the reported measurement results.
\begin{theorem}[Privacy against eavesdroppers]\label{sectionSe1}
Outside eavesdropper can not learn information about the participants' private sets through common attacks.
\end{theorem}
\begin{proof}
Suppose that an outside eavesdropper Eve attempts to learn information about the private sets of the participants from the transmitted photon sequences. However, every transmission is protected by randomly inserted decoy photons chosen from
\(\{|0\rangle,|1\rangle,|+\rangle,|-\rangle\}\), whose positions and preparation bases are unknown to Eve. We consider two representative attacks:

\emph{Intercept–Resend Attack.}
If Eve intercepts a transmitted sequence, measures each photon with a guessed basis, and resends a fabricated sequence, Eve inevitably disturbs a fraction of the decoy photons during her eavesdropping attempt. For each decoy photon, she avoids detection with a maximum probability of $3/4$. Consequently, if $\delta$ decoy photons are transmitted, the probability of Eve bypassing the eavesdropping check is at most $(3/4)^{\delta}$, while the detection probability is at least $1-(3/4)^{\delta}$. 

\emph{Auxiliary-qubit attack.}
Eve may also attach an auxiliary qubit \(|0\rangle_a\) to a transmitted decoy photon \(|\psi\rangle_d\) and apply an oracle-type operation
\(
U_f|x\rangle_d|y\rangle_a
=
|x\rangle_d|y\oplus f(x)\rangle_a .
\)
For computational-basis states, one has
\[
\begin{aligned}
\qquad \qquad U_f|0\rangle_d|0\rangle_a &= |0\rangle_d|f(0)\rangle_a, \\[5pt]
U_f|1\rangle_d|0\rangle_a &= |1\rangle_d|f(1)\rangle_a.
\end{aligned}
\]
Thus, if the decoy photon were known to be in the \(Z\)-basis, Eve might try to encode basis information into the auxiliary qubit. 
However, when the decoy photon is prepared in the \(X\)-basis, namely
\(|\psi\rangle_d\in\{|+\rangle,|-\rangle\}\), we have
\[
\begin{aligned}
&U_f|\pm\rangle_d|0\rangle_a
=
\frac{1}{\sqrt{2}}
\left(
U_f|0\rangle_d|0\rangle_a
\pm
U_f|1\rangle_d|0\rangle_a
\right)\\
&=
\frac{1}{\sqrt{2}}
\left(
|0\rangle_d|f(0)\rangle_a
\pm
|1\rangle_d|f(1)\rangle_a
\right)\\
&=
\frac{1}{\sqrt{2}}
\left[
|+\rangle_d
\frac{|f(0)\rangle_a\pm|f(1)\rangle_a}{\sqrt{2}}
+
|-\rangle_d
\frac{|f(0)\rangle_a\mp|f(1)\rangle_a}{\sqrt{2}}
\right].
\end{aligned}
\]
the reduced state of Eve's auxiliary qubit is identical for the two possible decoy states. 
Hence, Eve cannot distinguish \(|+\rangle\) from \(|-\rangle\) from the auxiliary system. 
Furthermore, Since Eve knows neither the decoy positions nor their bases of $|\psi\rangle_d$, the auxiliary-qubit attack cannot reveal useful information without being exposed by the decoy check.
\end{proof}
\begin{theorem}[Privacy against TP]\label{sectionSe2}
TP learns neither the participants' private sets, nor the intersection indices, nor the exact intersection cardinality.
\end{theorem}
\begin{proof}
(1) \emph{TP cannot learn any participants' private sets.}
To recover the private set \(X_i\), TP would have to determine the encoded membership vector \(Y_i\), or equivalently, distinguish whether \(P_i\) applies the data-dependent rotation \(\vartheta_{i,t}=0\) or \(\vartheta_{i,t}=\pi/n\) at each position \(t\).

TP may attempt to prepare entangled photons instead of single photons, keep one subsystem, and send the other subsystem through the protocol. 
Let \(\rho_{AB}\) be such an entangled state, where subsystem \(A\) is retained by TP and subsystem \(B\) is sent to \(P_i\). 
If \(P_i\) applies a local unitary \(U_i\) on subsystem \(B\), then the reduced state available to TP is
\[
\operatorname{Tr}_{B}\!\left[(I\otimes U_i)\rho_{AB}(I\otimes U_i^\dagger)\right]
=
\operatorname{Tr}_{B}(\rho_{AB}).
\]
Hence, the subsystem kept by TP is independent of the rotation applied by \(P_i\). TP may also intercept the sequence immediately after \(P_i\) and attempt to distinguish whether \(P_i\) applies \(\vartheta_{i,t}=0\) or \(\vartheta_{i,t}=\pi/n\) at position \(t\) via unambiguous state discrimination (USD) test \cite{USD,Chefles98}. 
Let \(\lvert\phi_{i,t}\rangle\) be the photon received by \(P_i\). 
The two possible output states are
\[
\begin{aligned}
\lvert\phi_0\rangle
&=R_y(T_{i,t}+\Delta_{i,t})\lvert\phi_{i,t}\rangle,\\
\lvert\phi_1\rangle
&=R_y\!\left(\frac{\pi}{n}+T_{i,t}+\Delta_{i,t}\right)\lvert\phi_{i,t}\rangle .
\end{aligned}
\]
For any fixed incoming state, they can be written as
\[
\begin{aligned}
\lvert\phi_0\rangle
&=\cos(\omega)\lvert0\rangle+\sin(\omega)\lvert1\rangle,\\
\lvert\phi_1\rangle
&=\cos\!\left(\omega+\frac{\pi}{2n}\right)\lvert0\rangle
+\sin\!\left(\omega+\frac{\pi}{2n}\right)\lvert1\rangle,
\end{aligned}
\]
where \(\omega\) collects all rotations common to the two cases. 
For two nonorthogonal pure states, the optimal USD success probability is governed by their overlap~\cite{Chefles98}. Hence,
\[
\mathsf{Prob}^{\mathrm{USD}}
=
1-\left|\langle \phi_0 \mid \phi_1\rangle\right|
=
1-\cos\left(\frac{\pi}{2n}\right).
\]
This bound already assumes that TP knows the incoming state and the common rotation offset. 
In the actual protocol, TP does not know \(\Delta_{i,t}\); moreover, the encoded vector \(Y_i\) is generated after the hiding map \(x\mapsto kx\bmod M\), where \(k\) is unknown to TP; 
For \(i>1\), the incoming photons are further masked by preceding participants. 
Hence, TP cannot reliably infer \(Y_{i,t}\), and therefore cannot recover \(P_i\)'s private set.

(2) \emph{TP cannot interpret \(z^{\mathsf S}\), \(z^{\mathsf O}\).} After the measurement phase, the semantic meaning of these labels is determined by the participant-side flip vector \(\mathbf b\), which is realized through the correlated aggregate rotations \(\{\Delta_{i,t}\}_{i=1}^n\) and is unknown to TP. 
Moreover, TP does not know the hidden real-domain selector \(m_{\mathcal U}\), the anchor selector \(m_{\mathcal A}\), or the position-hiding key \(k\). 
Hence, TP can neither map measurement labels to actual intersection coordinates nor distinguish between real-domain and anchor positions.

(3) \emph{TP cannot recover the cardinality.} Since TP cannot interpret \(z^{\mathsf S}\) and \(z^{\mathsf O}\), it cannot derive the exact intersection cardinality directly from the quantum measurement results. 
The classical cardinality-testing phase does not reveal it either: the OIP calls provide only additive shares of the required inner products, and the garbled circuit outputs only the predicate bit \(\mathsf{flag}\). At most, by counting the deterministic outcomes represented by \(z^{\mathsf S}\) and \(z^{\mathsf O}\), TP may infer the aggregate
$\biggl| \left|\bigcap_{i=1}^{n}X_i\right| + \left|\mathcal U\setminus\bigcup_{i=1}^{n}X_i\right| + |\mathcal A^+| + |\mathcal A^-| \biggr|$ up to the statistical classification error controlled by \(\ell\) which is insufficient to determine the exact intersection cardinality.
\end{proof}
\begin{remark}[Anchor-based tamper detection]
Although TP is assumed to be semi-honest, the anchor positions provide a lightweight consistency check against modification of the reported vectors \(z^{\mathsf S}\) and \(z^{\mathsf O}\). 
Suppose that TP modifies the reported measurement labels at \(r\) distinct positions. Any nontrivial modification on an anchor position violates the anchor-consistency condition \(d_{\mathcal A}=0\). 
Hence, TP avoids detection only if all modified positions fall inside the hidden real domain \(k\mathcal U\) with \(\mathsf{Prob}^{\text{Undetected}}={\binom{|\mathcal U|}{r}}/{\binom{M}{r}}.\)
\end{remark}
\begin{theorem}[Privacy against other participants]\label{sectionSe3}
No participant learns information about another participant's private set even in collusion with other participants.
\end{theorem}
\begin{proof}
We consider an target participant \(P_i\) and assume, as TP does not collude with any participant.  To recover \(X_i\), a dishonest participant or a coalition of participants must infer the encoded membership bit \(Y_{i,t}\) at each position \(t\), namely, whether \(P_i\) applies \(\vartheta_{i,t}=0\) or \(\vartheta_{i,t}=\pi/n\).

\emph{Direct observation by the next participant.}
Consider first the participant \(P_{i+1}\), who receives the photon sequence immediately after \(P_i\). 
At position \(t\), the outgoing state of \(P_i\) is of the form
\(
R_y(\vartheta_{i,t}+T_{i,t}+\Delta_{i,t})\lvert\phi_{i,t}\rangle.
\)
The next participant does not know \(\lvert\phi_{i,t}\rangle\), nor the TP--\(P_i\) secret rotation \(T_{i,t}\).  In particular, for \(i=1\), the state sent to \(P_1\) is already blinded by TP's initial rotation \(\vartheta_{0,t}\). 

\emph{Collusion of neighboring participants.}
Now consider a stronger attack in which \(P_{i-1}\) and \(P_{i+1}\) collude to recover \(P_i\)'s private set.  They may even replace the legitimate incoming photon by a known state \(|\varphi_t\rangle\) before sending it to \(P_i\), and then inspect the state returned by \(P_i\). In this case, the security argument is essentially the same as the USD-based analysis in Part~(1) of Section~\ref{sectionSe2}. The only difference is that the unknown blinding rotation is now the \(T_{i,t}\) rather than \(\Delta_{i,t}\). Hence, the same USD bound applies. 

The same argument extends to any larger coalition of participants: regardless of how many other participants conspire, the target participant's TP-shared mask \(T_i\) remains unavailable to them. 
Therefore, they cannot recover \(Y_i\), and hence cannot obtain the private set \(X_i\).
\end{proof}

\begin{table*}[t]
\centering
\caption{Rotation vectors used in the toy-model simulation.}
\label{rotation-vectors}
\renewcommand{\arraystretch}{1.5} % 增加行高以适应分式
\setlength{\tabcolsep}{3pt}
\begin{tabular}{c c c c c}
\toprule
Participant
& Data rotation \(\boldsymbol{\vartheta}_i\)
& Pairwise mask \(T_i\)
& Flip share \(\boldsymbol{\Delta}_i\)
& Total rotation \(\boldsymbol{\vartheta}_i+T_i+\boldsymbol{\Delta}_i\) \\
\midrule
\(P_1\)
& \((0, \frac{\pi}{3}, \frac{\pi}{3}, \frac{\pi}{3}, \frac{\pi}{3}, 0, 0, 0)\)
& \((\frac{\pi}{12}, \frac{\pi}{4}, \frac{5\pi}{12}, \frac{7\pi}{12}, \frac{3\pi}{4}, \frac{11\pi}{12}, \frac{\pi}{6}, \frac{\pi}{3})\)
& \((\frac{\pi}{6}, \frac{\pi}{4}, \frac{\pi}{3}, \frac{5\pi}{12}, \frac{\pi}{6}, \frac{\pi}{4}, \frac{\pi}{3}, \frac{5\pi}{12})\)
& \((\frac{\pi}{4}, \frac{5\pi}{6}, \frac{13\pi}{12}, \frac{4\pi}{3}, \frac{5\pi}{4}, \frac{7\pi}{6}, \frac{\pi}{2}, \frac{3\pi}{4})\) \\

\(P_2\)
& \((0, \frac{\pi}{3}, \frac{\pi}{3}, \frac{\pi}{3}, 0, 0, \frac{\pi}{3}, \frac{\pi}{3})\)
& \((\frac{\pi}{6}, \frac{\pi}{3}, \frac{\pi}{2}, \frac{2\pi}{3}, \frac{5\pi}{6}, 0, \frac{\pi}{4}, \frac{5\pi}{12})\)
& \((\frac{\pi}{3}, \frac{\pi}{6}, \frac{\pi}{4}, \frac{\pi}{6}, \frac{\pi}{3}, \frac{\pi}{6}, \frac{\pi}{4}, \frac{\pi}{6})\)
& \((\frac{\pi}{2}, \frac{5\pi}{6}, \frac{13\pi}{12}, \frac{7\pi}{6}, \frac{7\pi}{6}, \frac{\pi}{6}, \frac{5\pi}{6}, \frac{11\pi}{12})\) \\

\(P_3\)
& \((0, \frac{\pi}{3}, \frac{\pi}{3}, \frac{\pi}{3}, \frac{\pi}{3}, 0, 0, \frac{\pi}{3})\)
& \((\frac{\pi}{4}, \frac{5\pi}{12}, \frac{7\pi}{12}, \frac{3\pi}{4}, \frac{11\pi}{12}, \frac{\pi}{12}, \frac{\pi}{3}, \frac{\pi}{2})\)
& \((\frac{3\pi}{2}, \frac{7\pi}{12}, \frac{17\pi}{12}, \frac{5\pi}{12}, \frac{3\pi}{2}, \frac{7\pi}{12}, \frac{17\pi}{12}, \frac{17\pi}{12})\)
& \((\frac{7\pi}{4}, \frac{4\pi}{3}, \frac{7\pi}{3}, \frac{3\pi}{2}, \frac{11\pi}{4}, \frac{2\pi}{3}, \frac{7\pi}{4}, \frac{9\pi}{4})\) \\
\bottomrule
\end{tabular}
\end{table*}

\section{Simulation and Performance Evaluation}\label{sectionSimulation}
\subsection{Experimental Simulation}
\noindent\emph{\textbf{Toy Model.}}
To illustrate the quantum phase of the proposed protocol, we consider a toy instance with \(n=3\) participants.  The original domain is \(\mathcal U=\{0,1,2,3,4,5\}\), and we introduce one positive anchor and one negative anchor, \(\mathcal A^+=\{6\}\) and \(\mathcal A^-=\{7\}\).  Hence, \(\widetilde{\mathcal U}=\mathcal U\cup\mathcal A^+\cup\mathcal A^-\) and \(M=|\widetilde{\mathcal U}|=8\). We set the threshold as \(\tau=2\). 
The private sets of the three participants are \(X_1=\{1,3,4\},X_2=\{1,2,3,5\},X_3=\{1,3,4,5\},
\) whose real intersection is \(X_1\cap X_2\cap X_3=\{1,3\}\).

Each participant appends the positive anchor set and apply the position-hiding map with hiding key \(k=3\in\mathbb Z_8^*\). This yields the position
\(
k\widetilde{\mathcal U}=\{0,1,2_{\mathcal A^+},3,4,5_{\mathcal A^-},6,7\},
\)
and
\(
k\widetilde X_1=\{1,2,3,4\},
k\widetilde X_2=\{1,2,3,6,7\},
k\widetilde X_3=\{1,2,3,4,7\}.
\)
Accordingly, the encoded indicator vectors are
\(
Y_1=(0,1,1,1,1,0,0,0),
Y_2=(0,1,1,1,0,0,1,1),
Y_3=(0,1,1,1,1,0,0,1).
\)

For the hidden semantic flip, we fix \(\mathbf b=(0,1,0,1,0,1,\allowbreak 0,0),\)
which determines the aggregate flip angles $\boldsymbol\Theta=(0,\pi,0,\pi,0,\pi,0,0) \pmod {2\pi}.$ The resulting participant-side reference label vector is \(\rho=(1,0,1_{\mathcal A^+},0,1,1_{\mathcal A^-},1,1).\)
The corresponding correlated aggregate rotation shares are chosen as 
$\Delta_1 = (\frac{\pi}{6}, \allowbreak \frac{\pi}{4}, \allowbreak \frac{\pi}{3}, \allowbreak \frac{5\pi}{12}, \allowbreak \frac{\pi}{6}, \allowbreak \frac{\pi}{4}, \allowbreak \frac{\pi}{3}, \allowbreak \frac{5\pi}{12})$, 
$\Delta_2 = (\frac{\pi}{3}, \allowbreak \frac{\pi}{6}, \allowbreak \frac{\pi}{4}, \allowbreak \frac{\pi}{6}, \allowbreak \frac{\pi}{3}, \allowbreak \frac{\pi}{6}, \allowbreak \frac{\pi}{4}, \allowbreak \frac{\pi}{6})$, 
$\Delta_3 = (\frac{3\pi}{2}, \allowbreak \frac{7\pi}{12}, \allowbreak \frac{17\pi}{12}, \allowbreak \frac{5\pi}{12}, \allowbreak \frac{3\pi}{2}, \allowbreak \frac{7\pi}{12}, \allowbreak \frac{17\pi}{12}, \allowbreak \frac{17\pi}{12})$, 
which satisfy $\Delta_{1,t}+\Delta_{2,t}+\Delta_{3,t} \equiv \Theta_t \pmod{2\pi}$ for every $t \in \{0,\dots,7\}$. For the TP--participant pairwise rotation masks, we take 
$T_1 = (\frac{\pi}{12}, \allowbreak \frac{\pi}{4}, \allowbreak \frac{5\pi}{12}, \allowbreak \frac{7\pi}{12}, \allowbreak \frac{3\pi}{4}, \allowbreak \frac{11\pi}{12}, \allowbreak \frac{\pi}{6}, \allowbreak \frac{\pi}{3})$, 
$T_2 = (\frac{\pi}{6}, \allowbreak \frac{\pi}{3}, \allowbreak \frac{\pi}{2}, \allowbreak \frac{2\pi}{3}, \allowbreak \frac{5\pi}{6}, \allowbreak 0, \allowbreak \frac{\pi}{4}, \allowbreak \frac{5\pi}{12})$, 
$T_3 = (\frac{\pi}{4}, \allowbreak \frac{5\pi}{12}, \allowbreak \frac{7\pi}{12}, \allowbreak \frac{3\pi}{4}, \allowbreak \frac{11\pi}{12}, \allowbreak \frac{\pi}{12}, \allowbreak \frac{\pi}{3}, \allowbreak \frac{\pi}{2})$. 
TP further samples the initial blinding rotation $\boldsymbol\vartheta_0 = (\frac{\pi}{12}, \allowbreak \frac{\pi}{6}, \allowbreak \frac{\pi}{4}, \allowbreak \frac{\pi}{3}, \allowbreak \frac{5\pi}{12}, \allowbreak \frac{\pi}{2}, \allowbreak \frac{7\pi}{12}, \allowbreak \frac{2\pi}{3})$ and prepares the initial photon sequence
\(
S_0=
\bigl(
|0\rangle,\ |+\rangle,\ |1\rangle,\ |-\rangle,\
|0\rangle,\ |+\rangle,\allowbreak \ |1\rangle,\allowbreak \ |-\rangle
\bigr).
\) The summary of rotation in our toy model is shown in Table \ref{rotation-vectors}.

\begin{figure*}[t]  
    \centering
    \includegraphics[width=0.8\linewidth]{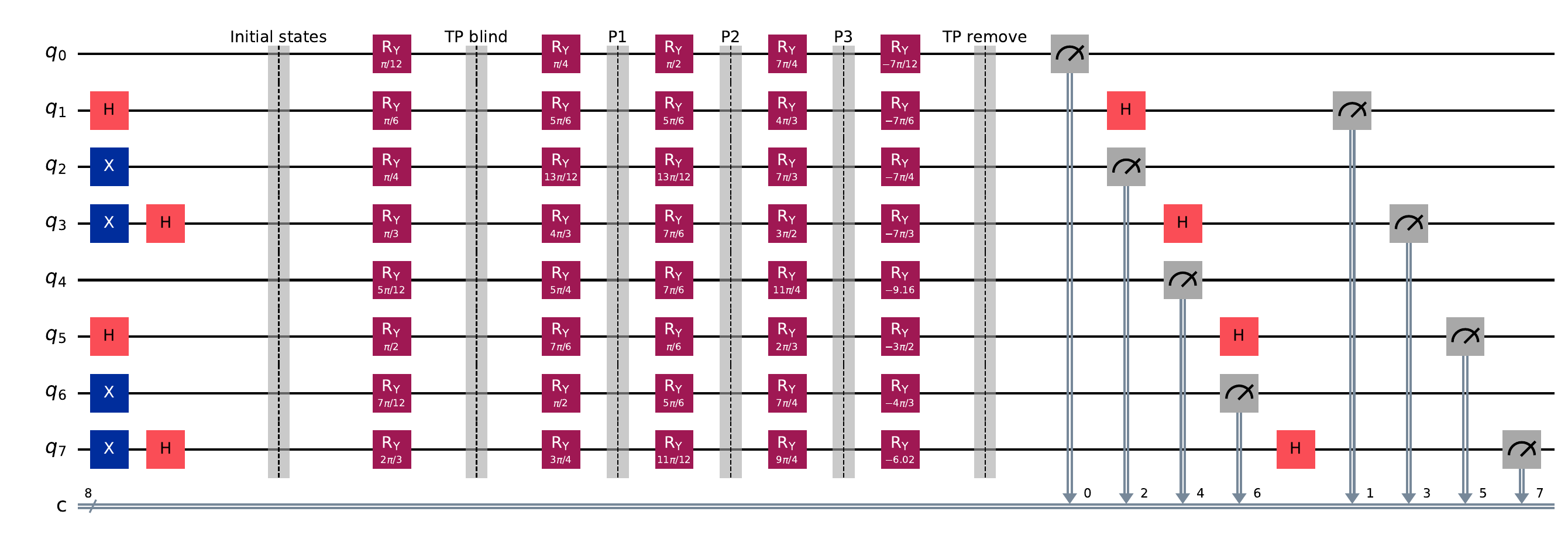}
    \caption{Quantum Circuit of toy model.}
    \label{figQuantumcircuit}
\end{figure*}
\begin{figure*}[t]  
    \centering
    \includegraphics[width=0.75\linewidth]{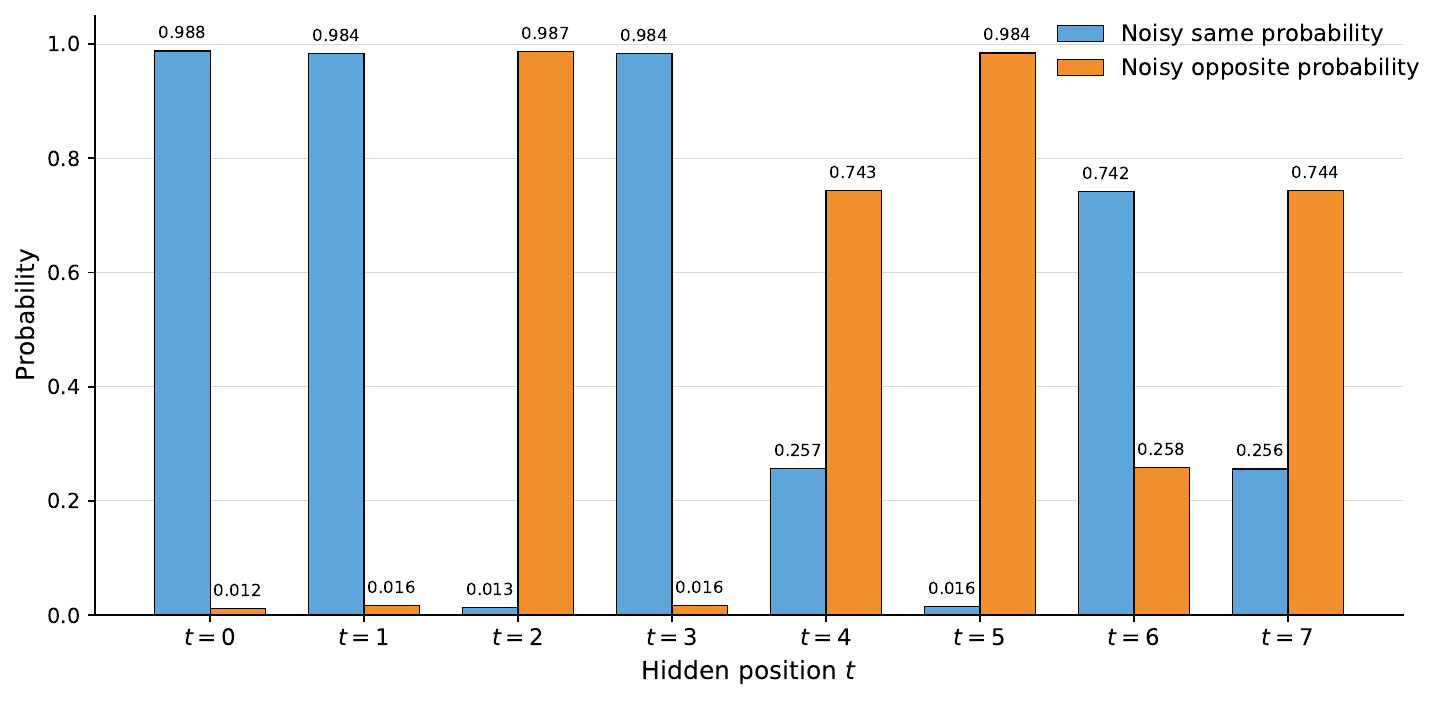}
    \caption{Noisy Probabilities.}
    \label{figProbabilities}
\end{figure*}
\begin{table}[t]
\centering
\small
\caption{Measurement results of the toy-model simulation.}
\label{measurement-results}
\renewcommand{\arraystretch}{1.2}
\setlength{\tabcolsep}{5pt}
\begin{tabular}{c c c c c c c}
\toprule
\multirow{2}{*}{\(t\)}
& \multirow{2}{*}{Origin}
& \multirow{2}{*}{Class}
& \multicolumn{2}{c}{Theory}
& \multicolumn{2}{c}{Noisy} \\
\cmidrule(lr){4-5}\cmidrule(lr){6-7}
& &
& Same
& Opp.
& Same
& Opp. \\
\midrule
0 & real(0) & all-same     & 1.000 & 0.000 & 0.988 & 0.012 \\
1 & real(3) & all-same     & 1.000 & 0.000 & 0.984 & 0.016 \\
2 & \(\mathcal{A}^+\) & all-opposite & 0.000 & 1.000 & 0.013 & 0.987 \\
3 & real(1) & all-same     & 1.000 & 0.000 & 0.984 & 0.016 \\
4 & real(4) & mixed        & 0.250 & 0.750 & 0.257 & 0.743 \\
5 & \(\mathcal{A}^-\) & all-opposite & 0.000 & 1.000 & 0.016 & 0.984 \\
6 & real(2) & mixed        & 0.750 & 0.250 & 0.742 & 0.258 \\
7 & real(5) & mixed        & 0.250 & 0.750 & 0.256 & 0.744 \\
\bottomrule
\end{tabular}
\end{table}

\textbf{\textsf{Qiskit}} \emph{\textbf{Simulation.}}
Based on the above toy model, we implement the quantum phase of the proposed protocol using \textsf{Qiskit}. 
The corresponding quantum circuit is shown in Fig.~\ref{figQuantumcircuit}. 
For each hidden position \(t\), the residual rotation after TP removes its initial blinding rotation and the pairwise masks is \(\frac{r_t\pi}{n}+b_t\pi,\) where \(r_t=\sum_{i=1}^{n}Y_{i,t}\) is the number of participants holding the encoded element at position \(t\). 
Hence, the theoretical probabilities of obtaining the same and opposite measurement relations are
\[
p_t^{\mathsf S}
=
\cos^2\!\left(
\frac{r_t\pi/n+b_t\pi}{2}
\right),
\quad
p_t^{\mathsf O}
=
\sin^2\!\left(
\frac{r_t\pi/n+b_t\pi}{2}
\right).
\]

To evaluate the robustness of the quantum phase, we perform noisy simulation by adding depolarizing noise with rate \(0.2\%\), phase-damping noise with rate \(0.4\%\), and readout error with rate \(0.5\%\) which can be realized through \textsf{Qiskit\_aer.noise}. These channels are chosen to reflect representative gate, dephasing, and measurement imperfections in the simulated quantum execution.
The resulting noisy same/opposite probabilities are plotted in Fig.~\ref{figProbabilities}, and their comparison with the theoretical values is summarized in Table~\ref{measurement-results}. 

Using \(0.9\) as the acceptance threshold for deterministic classification, the  noisy simulation yields
\(
z_{\mathrm{sim}}^{\mathsf S} = (1,1,0,1,0,0,0,0),z_{\mathrm{sim}}^{\mathsf O} =(0,0,1,0,0,1,0,0),
\)
which exactly agree with the theoretical expectation. 
Therefore, when these simulated outputs are supplied to the classical subprotocol 
\hyperref[subprotocol]{\(\Pi_{\mathrm{CT}}^{\tau}\)}, the anchor-consistency check is satisfied and the threshold condition is accepted. 
The participants can then reconstruct the real intersection as \(X_1\cap X_2\cap X_3=\{1,3\}.\)

\subsection{Performance and Comparison}
\noindent\emph{\textbf{Quantum Communication Cost.}}
In the quantum phase, TP prepares \(\ell\) photon sequences of length \(M\), i.e., \(\ell M\) signal photons in total. 
These sequences are transmitted along the chain
\(\mathrm{TP}\rightarrow P_1\rightarrow\cdots\rightarrow P_n\rightarrow\mathrm{TP}\).
If \(\delta\) decoy photons are inserted in each transmission for eavesdropping detection, then the total quantum communication cost is
\(
\mathcal{O}\bigl((n+1)(\ell M+\delta)\bigr)
=
\mathcal{O}\bigl(n(\ell M+\delta)\bigr).
\)

\emph{\textbf{Quantum Computation Cost.}}
TP applies \(\ell M\) initial blinding rotations, \(\ell M\) mask-removal rotations, and performs \(\ell M\) final measurements on the signal photons. 
Hence, TP's quantum computation cost is \(\mathcal{O}(\ell M+\delta)\). 
Each participant \(P_i\) performs \(\ell M\) local rotation operations and \(\mathcal{O}(\delta)\) decoy-photon measurements, giving a per-participant cost of \(
\mathcal{O}(\ell M+\delta).\)
Therefore, the total quantum computation cost of all participants is \( \mathcal{O}\bigl(n(\ell M+\delta)\bigr).\) Notice that the anchor sets are chosen with size independent of the asymptotic domain scale, we have \(M=\mathcal{O}(q)\).

\emph{\textbf{Performance Comparison.}}
Table~\ref{comparison} compares our protocol with representative classical and quantum TPSI schemes. 
Unlike the classical constructions in \cite{BMRR21,Zhang21,Hu24}, our protocol is quantum-cryptography based and therefore fits the long-term security motivation of privacy-preserving computation in the quantum setting. 
At the same time, our construction relies only on single-photon rotations in the quantum phase and a lightweight OLE-based classical post-processing step. 
This differs from the QHE-based design of \cite{Wang26}, whose encrypted quantum evaluation introduces Clifford gates. 

More importantly, our protocol explicitly realizes cardinality testing: it outputs the predicate
\(
\mathbf 1\!\left[\left|\bigcap_{i=1}^{n}X_i\right|\ge \tau\right],
\)
rather than first exposing the exact intersection cardinality \(\left|\bigcap_{i=1}^{n}X_i\right|\) and then comparing it with \(\tau\). 
This distinction is essential for TPSI specially under the TP setting and makes our functionality consistent with the classical TPSI paradigm represented by \cite{BMRR21,Zhang21,Hu24}. 
By contrast, in the existing quantum TPSI protocols \cite{Mohanty24,LI26,Wang26}, TP remains responsible not only for measurement, but also for interpreting the measurement outcomes and determining the intersection-related result. 
Our protocol separates these roles: TP is reduced to a blinded measurement party, while the threshold decision is delegated to the classical subprotocol \hyperref[subprotocol]{\(\Pi_{\mathrm{CT}}^{\tau}\)}.

\begin{table*}[t]
  \centering
  \scriptsize
  \caption{Comparison of TPSI protocols.}
  \label{comparison}

  \setlength{\tabcolsep}{3pt}
  \renewcommand{\arraystretch}{0.95}

  \resizebox{\textwidth}{!}{
  \begin{tabular}{cccccccc}
    \toprule
    Protocol
      & \begin{tabular}[c]{@{}c@{}}Party\\Scenarios\end{tabular}
      & \begin{tabular}[c]{@{}c@{}}Quantum\\Cryptography Based\end{tabular}
      & Technologies
      & \begin{tabular}[c]{@{}c@{}}TP\\Model\end{tabular}
      & \begin{tabular}[c]{@{}c@{}}Cardinality\\Testing\end{tabular}
      & \begin{tabular}[c]{@{}c@{}}Communication\\Complexity\end{tabular}
      & \begin{tabular}[c]{@{}c@{}}Computation\\Complexity\end{tabular} \\
    \midrule
      \cite{BMRR21}&  Multi-Party & No & TFHE & None & Yes & $\mathcal{O}\bigl(n\tau\bigl)$& $\mathcal{O}\bigl(n-\tau\bigl)^4$\\
      \cite{Zhang21}&Two-Party & No & \begin{tabular}[c]{@{}c@{}}GBF\\OT Extension\end{tabular} & None &Yes& $\mathcal{O}\bigl(\lambda B\bigr)$&$\mathcal{O}\bigl(n\log n\bigr)$\\
      \cite{Hu24}& Two-Party & No & FHE & None & Yes & $\mathcal{O}\bigl(q\sqrt{S})$ & $\mathcal{O}\bigl(qSNd_{ks}B)$\\
      \cite{Mohanty24}& Two-Party & Yes & Rotation & Trusted & No & $\mathcal{O}\bigl(q+(\ell+\delta)q\bigr)$ & $\mathcal{O}\bigl(q+(\ell+\delta)q\bigr)$\\
      \cite{LI26}& Multi-Party& Yes & Rotation & Semi-honest & No & $\mathcal{O}\bigl(nQ+nq(\ell+\delta)\bigr)$ &$\mathcal{O}\bigl(nQ+nq(\ell+\delta)\bigr)$\\
      \cite{Wang26} & Multi-Party & Yes & QHE & Trusted & No & $\mathcal{O}\bigl(n(q+\delta)\bigr)$ & $\mathcal{O}\bigl(q(n+\delta)\bigr)$ \\
      Ours& Multi-Party & Yes 
      & \begin{tabular}[c]{@{}c@{}}Rotation\\OLE\end{tabular} 
      & Semi-honest & Yes &$\mathcal{O}\bigl(n(\ell q+\delta)\bigr)$ & $\mathcal{O}\bigl(n(\ell q+\delta)\bigr)$ \\
    \bottomrule
  \end{tabular}
  }

  \vspace{0.4em}
  \begin{minipage}{0.98\textwidth}
    \footnotesize
    \emph{Note:}
    \(n\) : number of participating parties;
    \(q\) : data size of each participant;
    \(\ell\) : number of repeated sequences;
    \(\delta\) : number of decoy positions;
    \(Q\) : number of initial qubits prepared by SA-OQKD;
    \(B\) : number of hash functions in the Bloom filter;
    \(N\) : lattice dimension of LWE;
    \(d_{ks}\) denotes a small constant associated with the KeySwitch technique;
    \(S\) : size of the server-side dataset;
    and \(\lambda\) : security parameter.
  \end{minipage}
\end{table*}
\section{Conclusion}\label{sectionConclusion}
\noindent\emph{Conclusion.} This paper presented a rotation based MP-QTPSI protocol with explicit cardinality testing. 
Unlike existing quantum TPSI constructions that rely on TP to recover the intersection cardinality before threshold comparison, our protocol decouples TP's measurement role from the threshold-testing functionality. 
By combining hidden-label quantum measurements, participant-side correlated rotations, and a classical OLE-assisted testing procedure, the proposed protocol realizes
\(
\mathbf 1\!\left[\left|\bigcap_{i=1}^{n}X_i\right|\ge \tau\right]
\)
without exposing the exact intersection cardinality. 

An additional advantage of the hidden-label testing framework is its flexibility. 
Since the final decision is reduced to a classical predicate over secret-shared consistency statistics, the present threshold rule can be naturally extended to richer policies, such as weighted threshold, ratio-based threshold \cite{BMRR21}, or over-threshold \cite{Arpaci25} without changing the core quantum interaction structure.

\emph{Future Work.} (1) it is meaningful to extend the current semi-honest TP model to stronger adversarial settings, including malicious TP behavior and possible collusion between TP and participants; 
(2) it is worth investigating whether the quantum communication and computation costs can be made threshold-dependent, analogous to classical TPSI designs such as~\cite{GhoshSimkin19}, rather than scaling with the full encoded domain length;
(3) although current quantum PSI-type protocols do not require participants to hold sets of equal cardinality, they typically encode inputs into equal-length indicator vectors, which still reflects a balanced-PSI style representation. 
Developing genuinely unbalanced quantum PSI may further reduce quantum resources and improve practicality in highly skewed data regimes.

\bibliographystyle{myalpha-hyper-tail} 
\bibliography{myrefs}       
\end{document}